\documentclass[a4paper,10pt]{article}

\usepackage[utf8]{inputenc}
\usepackage{istgame}
\usepackage{tikz}
\usepackage{xparse}
\usepackage{expl3}
\usepackage{amsmath,amssymb, amsfonts,amsthm}
\usepackage{mathtools}
\usepackage{enumerate}
\usepackage{enumitem}
\usepackage{natbib}
\usepackage{appendix}
\usepackage{bm}
\usepackage{geometry}
\usepackage{hyperref}
\usepackage{cleveref}
\usepackage[misc]{ifsym}
\usepackage[singlespacing]{setspace}
\usepackage[e]{esvect}
\usepackage{nicefrac}
\usepackage{thmtools}
\usepackage{soul}
\usepackage{xcolor}
 \newcommand{\Acc}[2]{\mathcal{A}_{#2}}

\usepackage{todonotes}

\usepackage{cancel}
\usepackage{shellesc}
\usepackage{pgfplots}
\pgfplotsset{compat=1.14}
\hypersetup{
	colorlinks = true,
	urlcolor = blue,
	linkcolor = red,
	citecolor = blue}

\usepackage{appendix}
\usepackage{chngcntr}
\usepackage{apptools}
\AtAppendix{\counterwithin{lemma}{section}}
\AtAppendix{\counterwithin{theorem}{section}}
\AtAppendix{\counterwithin{remark}{section}}
\AtAppendix{\counterwithin{example}{section}}
\AtAppendix{\counterwithin{proposition}{section}}
\AtAppendix{\counterwithin{definition}{section}}
\AtAppendix{\counterwithin{corollary}{section}}

\newtheorem{theorem}{Theorem}

\newtheorem{lemma}{Lemma}

\newtheorem{definition}{Definition}
\newtheorem*{theorem*}{Theorem}

\theoremstyle{remark}
\newtheorem{remark}{Remark}

\newcommand{\defin}[1]{{\bfseries\upshape #1}}
\renewcommand{\le}{\leqslant}
\renewcommand{\ge}{\geqslant}
\renewcommand{\leq}{\leqslant}
\renewcommand{\geq}{\geqslant}

\newcommand{\Prob}{\mathbf{P}}
\newcommand{\A}{\mathcal{A}}

\newcommand{\cl}{\operatorname{cl}}

\newcommand{\f}{\ensuremath{{\mathcal{D}_\mathcal{A}}}}
\newcommand{\cog}{\mathcal{W}}
\newcommand{\D}{\ensuremath{\mathcal{D}}}

\newcommand{\bd}{\operatorname{bd}}
\newcommand{\interior}{\operatorname{int}}

\newcommand{\conv}{\mathrm{conv}}

\newcommand{\cone}{\operatorname{cone}}

\newcommand{\R}{\ensuremath{\mathbb{R}}}

\newcommand{\X}{\mathcal{X}}
\newcommand{\Pro}{\mathbf{P}}

\mathchardef\mhyphen="2D 

\newcommand*{\defcol}{\mathrel{\vcenter{\baselineskip0.5ex \lineskiplimit0pt
                     \hbox{\scriptsize.}\hbox{\scriptsize.}}}%
                     =}

\title{A note on the induction of comonotonic additive risk measures from acceptance sets}
\author{Samuel S. Santos\thanks{\textbf{Corresponding author}. \href{mailto:ssolgons@uwaterloo.ca}{ssolgons@uwaterloo.ca}. Room 3102, 200 University Avenue West, Waterloo, ON, Canada. The author thanks for the support from the Brazilian Coordination for the Improvement of Higher Education Personnel (CAPES) under grant number 88882.439088/2019-01. Declarations of interest: none.} \\ \small{University of Waterloo} \and Marlon R. Moresco\thanks{\hspace{0.1cm} \href{mailto:marlon.moresco@concordia.ca}{marlon.moresco@concordia.ca}. Room LB-901, J.W. McConnell Building (LB), 1400 De Maisonneuve Blvd. W., Montreal, QC, Canada. Declarations of interest: none.} \\ \small{Concordia University} \and Marcelo B. Righi\thanks{\hspace{0.1cm} \href{mailto:marcelo.righi@ufrgs.br}{marcelo.righi@ufrgs.br}. Av. João Pessoa, 52 - Centro Histórico, Porto
Alegre, RS, Brazil. The author thanks for the support from the Brazilian National Council for Scientific and Technological Development (CNPq) under grant number 302614/2021-4. Declarations of interest: none.}\\ \small{Federal University of Rio Grande do Sul} \and Eduardo Horta \thanks{\hspace{0.1cm} \href{mailto:eduardo.horta@ufrgs.br}{eduardo.horta@ufrgs.br}. Av. Bento Gonçalves, 9500 - Agronomia, Porto Alegre, RS, Brazil. The author thanks for the support from the Brazilian National Council for Scientific and Technological Development (CNPq) under grant number 438642/2018-0. Declarations of interest: none.}\\ \small{Federal University of Rio Grande do Sul}}

\date{}

\begin{document}
\maketitle

\begin{abstract}
 We present simple general conditions on the acceptance sets under which their induced monetary risk and deviation measures are comonotonic additive. We show that acceptance sets induce comonotonic additive risk measures if and only if the acceptance sets and their complements are stable under convex combinations of comonotonic random variables. A generalization of this result applies to risk measures that are additive for random variables with \textit{a priori} specified dependence structures, e.g., perfectly correlated, uncorrelated, or independent random variables.
\end{abstract}

\maketitle

\section{Introduction}

The notion of risk is rooted in two fundamental concepts: the potential for adverse outcomes and the variability in expected results. Traditionally, risk has been understood as a measure of dispersion, such as variance, in line with the second concept \citep{markowitz1952}. However, the occurrence of critical events has brought attention to tail risk measurement, exemplified by well-known measures like Value at Risk (VaR) and Expected Shortfall (ES), which account for the possibility of extreme events, thus incorporating the first concept. \cite{danielsson2001academic} and \cite{embrechts2014academic} are remarkable references in this regard.

This study investigates the relationship between acceptance sets and risk / deviation measures, focusing on the property of comonotonic additivity. Roughly speaking, two random variables are comonotonic if the variability of one never offsets the variability of the other, that is, they move in the same direction. A financial intuition of the property of comonotonic additivity is the following: joining two comonotonic positions provides neither diversification benefits nor brings harm to the portfolio. Comonotonic additivity occupies a central place in the theory of risk measures
(seminal papers in this regard are \cite{wang1997}, \cite{yaari1987dual}, \cite{kusuoka}, and \cite{acerbi2002spectral}).

Acceptance sets are criteria used by financial regulators to distinguish between permissible and impermissible positions held by financial firms. However, acceptance sets alone do not provide direct guidance on how to convert non-permissible positions into permissible ones. This is the role of risk measures, which assign extended real values to quantify the risk (usually the tail risk) of financial positions. For non-permissible financial positions, monetary risk measures indicate the minimum amount of cash addition or assets addition required to make these positions permissible. This idea goes back to \cite{artzner1999}. For a review, see chapter 4 of \cite{follmer2016stochastic}. On the other hand, deviation measures may not reflect tail risk, as they are designed to quantify deviation. \cite{rockafellar06} is a landmark work in the axiomatic study of deviation measures, and \cite{pflug2007modeling} provide a handbook treatment. In analogy to risk measures, \cite{moresco2023minkowski} associated deviation measures to acceptance sets, and showed that generalized deviation measures (in the sense of \cite{rockafellar06}) represent how much a position must be shrunk or deleveraged for it to become acceptable.  As further references on the topic, \cite{nendel2021decomposition} and \cite{righi2019composition} studied the connection between risk, deviation measures, and premium principles. Also, \cite{grechuk2009maximum} used deviation measures to define restrictions on problems of maximum entropy.

From an axiomatic point of view, the properties of a risk measure directly translate into attributes of its acceptance set. It is well known that a risk measure is law-invariant, convex, positive homogeneous, and star-shaped if and only if its acceptance set is law-invariant, convex, conic, and star-shaped. However, the literature has no correspondence for comonotonic additivity beyond an attempt in finite probability spaces from \cite{rieger2017characterization}. In fact, additivity in general was never approached, to the best of our knowledge, through the perspective of acceptance sets.   

The additivity of a risk measure means that it is just as risky to have two positions added together in the same portfolio as it is to have them separated. If there were some diversification benefits in holding them together, we would require the acceptance set and the risk measure to be convex. For a discussion on convexity, see \cite{dhaene2008too}, \cite{tsanakas2009split}, and \cite{rau2019bigger}. If it were more risky to hold them together, the risk measure would be concave. From the perspective of acceptance sets, it translates into requiring the acceptance set's complement to be convex. Now, if the risk of two positions is the same regardless of whether they are in the same portfolio or not, then a combination of the two aforementioned concepts emerges. In this case, the risk measure should be both convex and concave, and both the acceptance set and its complement should be convex.

It is well known that the only linear risk measure is the expectation, and in this case, the above rationale holds trivially because both the acceptance set and its complement are half-spaces. However, we are interested in the additive property for random variables with specific dependence structures, such as independent, uncorrelated, and mainly, comonotonic random variables; that is, we do not require the risk measure to be additive in its whole domain, but just for specific random variables which, under some criterion, neither provide diversification benefit nor harm. 

Our main results show that this connection occurs for monetary and deviation measures. While the concept of monetary and deviation measures are similar, the technical tools to obtain those results are significantly different. In fact, up until recently, there was no such thing as an acceptance set for deviation measures. \cite{moresco2023minkowski} established the notion of acceptance sets for deviation measures, to which a crucial property is positive homogeneity. Since the main focus of this study is comonotonic additivity, which is a stronger property than positive homogeneity, we will exclusively consider deviation measures that satisfy the former condition. We focus on monetary risk measures in \Cref{Risk}, and deviation measures in \Cref{deviations}.

Regarding basic notation, let $(\Omega, \mathcal{F},\Prob)$ be a probability space and $L^{0}\defcol L^{0}(\Omega, \mathcal{F},\Prob)$ the space of equivalence classes of random variables (under the $\Prob \mhyphen$a.s. relation) and $L^{\infty}\defcol L^{\infty}(\Omega, \mathcal{F},\Prob)=\{X \in L^{0}: \Vert X \Vert_{\infty}< +\infty \}$, where $\Vert X \Vert _{\infty}= \inf \{m \in \mathbb{R}: |X|<m\}$ for all $X \in L^{0}$. Equalities and inequalities must be understood in the $\Prob \mhyphen$a.s. sense. For generality, we work on a Hausdorff topological vector space $\X$ such that $ L^\infty \subseteq \mathcal{X} \subseteq L^0$. The elements $X \in \mathcal{X}$ represent discounted net financial payoffs. We adopt the identify $\mathbb{R}\equiv\{X \in \mathcal{X}:X=c\text{ for some }c \in \mathbb{R}\}$. 
For any subset $A \subseteq \X$, we denote $\conv (A)$, $\cone (A)$, $A^\complement$ the convex hull, conic hull, and complement of $A$, respectively. Also, for any two sets $A,B \subseteq \mathcal{X}$, we denote $A+B=\{X \in \mathcal{X}: X=Y+Z,Y \in A, Z \in B\}$. It is valid noticing that if $A$ is non-empty, $0 \in \cone(A)$.  
Further, two random variables $X$ and $Y$ are \defin{comonotonic} if
\begin{equation}\label{def.comon}
    (X(\omega)-X(\omega'))(Y(\omega)-Y(\omega'))\ge 0\quad \Prob\otimes \Prob\mhyphen a.s. 
\end{equation}

The concept of comonotonicity dates back at least to \cite{hardy1934inequalities}. \cite{ruschendorf2013mathematical} and \cite{dhaene2020comonotonic} present further characterizations of comonotonic random variables.

\section{Monetary risk measures}\label{Risk}

We begin with some terminology on acceptance sets and monetary risk measures.

\begin{definition}\label{def.accept}
A nonempty set $\mathcal{A}\subseteq \mathcal{X}$ is called an acceptance set. It is a \defin{monetary acceptance set} if satisfies the following:
\begin{enumerate}[series = axioms.A,label=\textbf{\Alph*.}]
    \item (Monotonicity) $\mathcal{A}$ is \defin{monotone} if $X \in \mathcal{A}$ and $X\le Y$ implies $Y \in \mathcal{A}$.
     \item (Normalization) $\mathcal{A}$ is \defin{normalized} if $\inf\{m \in \mathbb{R}:m \in \mathcal{A}\}=0$.
\end{enumerate}
In addition, an acceptance set may fulfill: \begin{enumerate}[resume = axioms.A,label= \textbf{\Alph*.}]
   
    \item (Convexity) $\mathcal{A}$ is \defin{convex} if $\lambda \mathcal{A} + (1-\lambda)\mathcal{A} \subseteq \mathcal{A}$ whenever $\lambda \in [0,1]$.
  \end{enumerate}
We say that any set is \defin{comonotonic convex} if $X,Y \in \A $ implies  $\lambda X+ (1-\lambda )Y  \in \A$ for all comonotonic pairs $X,Y \in \mathcal{X}$.
\end{definition}

\begin{definition}\label{def.risk.measure}
A functional $\rho:\mathcal{X}\rightarrow \mathbb{R}\cup\{\infty\}$ is called a \defin{risk measure} if it satisfies:
\begin{enumerate}[series = axioms]
    \item (Monotonicity) $\rho$ is \defin{monotone} if $\rho(Y)\le \rho(X)$ whenever $X\le Y$ for $X,Y \in \mathcal{X}$.
    \item (Cash invariance) $\rho$ is \defin{cash invariant} if $\rho(X+m)=\rho(X)-m$ for any $X \in \mathcal{X}$ and $m \in \mathbb{R}$.
      \item (Normalization) $\rho$ is \defin{normalized} if $\rho(0)=0$.
\end{enumerate}
In addition, a functional may fulfil the following for some set $C\subseteq\X$:
\begin{enumerate}[resume = axioms]
     \item (Convexity) $\rho$ is \defin{convex in $C$} if $\rho(\lambda X + (1-\lambda)Y) \le \lambda \rho(X)+(1-\lambda)\rho(Y)$ for all $\lambda \in [0,1]$ and $X,Y \in C$.
   
        \item (Concavity) $\rho$ is \defin{concave in $C$} if $\rho(\lambda X + (1-\lambda)Y) \geq \lambda \rho(X)+(1-\lambda)\rho(Y)$ for all $\lambda \in [0,1]$ and $X,Y \in C$.
          \item (Additivity) $\rho$ is \defin{additive in $C$} if $\rho( X + Y) = \rho(X)+\rho(Y)$ for all $X,Y \in C$.
\end{enumerate}
If $C = \mathcal{X}$, we simply refer to the functional as convex, concave or additive. If $\rho$ is convex/concave/additive for comonotonic pairs, then we say it is comonotonic convex/concave/additive.

\end{definition}

\begin{remark}
    Since $0$ is comonotonic to any $X \in \mathcal{X}$, it is easy to see that, if $\rho$ is comonotonic convex, then $\rho(\lambda X)\le \lambda \rho(X)$ for $\lambda \in [0,1]$ and $\rho(\lambda X)\ge \lambda\rho(X)$ for $\lambda>1$. Risk measures satisfying this property are called \textit{star-shaped}. For theory and applications of star-shaped risk measures, see \cite{castagnoli2022star}, \cite{righi2021star}, \cite{righi2022star}, and \cite{moresco2022link}.
\end{remark}

\begin{definition}\label{def.induced}
Let $\rho$ be a risk measure and $\mathcal{A}$ a monetary acceptance set.
\begin{enumerate}
    \item The acceptance set induced by $\rho$ is defined as
\begin{equation}
    \mathcal{A}_{\rho}\defcol\{X \in \mathcal{X}:\rho(X)\le 0\}.
\end{equation}
    \item The risk measure induced by $\mathcal{A}$ is defined as
    \begin{equation}\label{risk.measure}
  \rho_{\mathcal{A}}(X)\defcol\inf\{m \in \mathbb{R}:X+m \in \mathcal{A}\}, \:\forall\: X \in \mathcal{X}.
\end{equation}
\end{enumerate}
\end{definition}

As shown, for instance, in \cite{artzner1999}, \cite{cheridito2009risk}, and \cite{kaina2009convex}, there exist direct links between acceptance sets and risk measures. The following relations between risk measures and acceptance sets will be used throughout the paper:
\begin{lemma}\label{lemma.basic}
    (Propositions 4.6 - \cite{follmer2016stochastic}; Lemma 2.5 - \cite{farkas2014beyond}) Let $\rho$ be a risk measure and let $\mathcal{A}$ be a monetary acceptance set. Then we have the following:
    \begin{enumerate}
        \item \label{lemma.basic.item.1} $\rho(X)=\rho_{\mathcal{A_{\rho}}}(X)$ for all $X \in \mathcal{X}$.
        \item \label{lemma.basic.item.2} $\{ X \in \X : \rho_\A (X) <0 \} \subseteq \A \subseteq \mathcal{A}_{\rho_{\mathcal{A}}} \subseteq \cl (\A)  $, where $\cl (\A)$ denotes the closure of $\A$.
         \item \label{lemma.basic.item.2b} If $\mathcal{A}$ is convex, then $\rho_{\mathcal{A}}$ is convex. Conversely, if $\rho$ is convex, then $\mathcal{A}_{\rho}$ is convex.
    \end{enumerate}
\end{lemma}

We use the following auxiliary function towards our way to this section's main result. Notice that it corresponds to the smallest upper bound for the amount of cash that can be added to some position without making it acceptable. 
\begin{definition}
Let $\mathcal{A} \subseteq \X$ be an acceptance set, then the functional $\psi_{\mathcal{A}^\complement} : \mathcal{X} \rightarrow \mathbb{R}\cup\{-\infty,+\infty\}$ induced by $\mathcal{A}^{\complement}$ be defined as
\begin{equation}
    \psi_{\mathcal{A}^{\complement}}(X)\defcol\sup\{m \in \mathbb{R}:X+m \in \mathcal{A}^{\complement}\},\: \forall\:X \in \mathcal{X}.
\end{equation}
\end{definition}

\begin{lemma}\label{risk.corisk}
Let $\mathcal{A}$ be a monetary acceptance set. Then $\rho_{\mathcal{A}}(X)=\psi_{\mathcal{A}^{\complement}}(X)$ for all $X \in \mathcal{X}$. 
\end{lemma}
\begin{proof}

From the monotonicity of monetary acceptance sets we have, for any $X \in \X$, that the real sets $\{ m \in \R : X +m \in \A \}$ and $\{ m \in \R : X +m \in \A^\complement \}$ are intervals that partition the real line. Hence, it follows that $\psi_{\mathcal{A}^{\complement}}(X) = \sup\{m \in \mathbb{R}:X+m \in \mathcal{A}^{\complement}\} = \inf \{ m \in \R : X +m \in \A \} = \rho_\A(X)$.
\end{proof}


The next result gives us sufficient conditions to induce convex, concave and additive risk measures. As formally stated in \Cref{def.accept2}, a set $C\subseteq \mathcal{X}$ is stable under scalar addition if $C+\mathbb R =C$.

\begin{theorem}\label{main.lemma}
Let $\mathcal{A}$ be a monetary acceptance set and $C \subseteq \X$ be stable under scalar addition.
\begin{enumerate}
    \item \label{main.lemma.s.1} If $\mathcal{A}\cap C $ is convex, then $\rho_{\mathcal{A}}$ is convex in $C$.
    \item \label{main.lemma.s.2} If $\mathcal{A}^{\complement} \cap C $ is convex, then  $\rho_{\mathcal{A}}$ is concave in $C$. 
    \item \label{main.lemma.s.3} If $\mathcal{A}^{\complement} \cap C $ and $\mathcal{A}\cap C $ are convex, then $\rho_{\mathcal{A}}$ is additive in $C$. 
\end{enumerate}
Furthermore, the converse implications hold if $\mathcal{A}$ is closed and $C$ is convex.
\end{theorem}
\begin{proof}

For \Cref{main.lemma.s.1}, let $X, Y \in C$ and note that there is $x , y \in  \R$ such that $X+x \in  \A$ and $Y + y \in \A$. As $C$ is stable under scalar addition, it also holds that $X + x \in  C $ for any $x \in \R$, and similarly for $Y+y$. Consequently, the convexity of $\A\cap C$ implies that $\lambda (X+x) + (1-\lambda)(Y+y) \in \A$ for any $\lambda \in [0,1]$. Therefore, $\rho_\A ( \lambda(X+x) + (1-\lambda)(Y+y)) \leq 0$, and the cash invariance of $\rho_\A$ implies $\rho_\A ( \lambda X+ (1-\lambda)Y) \leq \lambda x + (1-\lambda) y$. Then, taking the infimum over $x$ and $y$ yields
\[\rho_\A ( \lambda X+ (1-\lambda)Y) \leq \lambda \rho_{\A}(X) + (1-\lambda) \rho_{\A} (Y). \]

Regarding \Cref{main.lemma.s.2},  take $X, Y \in C$ and notice that, there is $x , y \in \R$ such that $X+x \in \A^\complement$ and $Y + y \in \A^\complement$. Therefore, the convexity of $\A^\complement\cap  C$ implies that $\lambda (X+x) + (1-\lambda)(Y+y) \in \A^\complement$ for any $\lambda \in [0,1]$. Hence we have $\rho_\A ( \lambda(X+x) + (1-\lambda)(Y+y)) > 0$, so the cash invariance of $\rho_\A$ implies $\rho_\A ( \lambda X+  (1-\lambda)Y) > \lambda x + (1-\lambda) y$. Then, taking a supremum over $x$ and $y$ and using \Cref{risk.corisk} yields 
\[ \rho_\A ( \lambda X+ (1-\lambda)Y) \geq \lambda \phi_{\A^\complement}(X) + (1-\lambda) \phi_{\A^\complement} (Y)=\lambda \rho_{\A}(X) + (1-\lambda) \rho_{\A} (Y).\]

For \Cref{main.lemma.s.3}, recall that under normalization, a map is linear if and only if it is both convex and concave. Hence, the claim is a direct consequence of the previous items.

When $\mathcal{A}$ is closed, the converse of \Cref{main.lemma.s.1} is straightforward by \Cref{lemma.basic}, \Cref{lemma.basic.item.2}. For \Cref{main.lemma.s.2}, take $X, Y \in \A^\complement \cap C$. Since $\mathcal{A}$ is closed, it holds that $\mathcal{A}=\mathcal{A}_{\rho_{\mathcal{A}}}$ (\Cref{lemma.basic.item.2} of \Cref{lemma.basic}), which implies $\rho_\A (X) >0 $ and $\rho_\A (Y)>0 $. Concavity of $\rho_\A$ in $C$ implies $\rho_\A ( \lambda X+ (1-\lambda)Y) \geq \lambda \rho_{\A}(X) + (1-\lambda) \rho_{\A} (Y) > 0 $, whence we conclude that $\lambda X+ (1-\lambda)Y \in  \A^\complement_{\mathcal{A}_{\rho}} = \A^{\complement}$. Additionally, it also belongs to $C$ as it is a convex set. Finally,  \Cref{main.lemma.s.3} follows by the previous items. 
\end{proof}

\begin{remark}
  Examples of sets $C \subseteq \mathcal{X}$ that fulfill the hypothesis of the above theorem are, for a given, fixed, $X \in \mathcal{X}$: the class of random variables independent of $X$, namely $C^{ind}_X = \{Y \in \mathcal{X}:X\text{ and }Y \text{ are independent}\}$, the set of random variables uncorrelated with $X$, that is $C^{uncor}_X \{Y \in \mathcal{X}:\operatorname{Cov}(X,Y)=0\}$, and the set of affine transformations of $X$, namely $C^{cov}_X \{Y \in \mathcal{X}:\operatorname{Cov}(X,Y)=1\}$. As an application of \Cref{main.lemma}, notice that, if $C^{ind}_X \cap \mathcal{A}$ and $C^{ind}_X \cap \mathcal{A}^{\complement}$ are convex for all $X \in \mathcal{X}$, then $\rho_\mathcal{A}$ is additive for independent random variables. This closely relates to the literature on additive risk measures and premium principles (see, for instance, \cite{goovaerts2004comonotonic}, and \cite{goovaerts2010note}).
\end{remark}

The preceding reasoning and results yield comonotonic additivity of $\rho$ whenever $\mathcal{A}$ and $\mathcal{A}^\complement$ are both convex for comonotonic pairs; this is the content of the main theorem in this section. We now show a result that relates comonotonic variables to the needed assumptions. To this end, we will denote by $C_X \coloneqq \{Y \in \X\colon\, \text{$Y$ is comonotonic to $X$} \} $ the set of all random variables that are comonotonic to $X \in \mathcal{X}$.

\begin{lemma}\label{lemma convex cone comono}
	Let $X\in\X$. The following holds:
    \begin{enumerate}
        \item \label{lemma convex cone comono.item1}$C_X$ is a convex cone that is closed with respect to the topology of convergence in probability.
        \item \label{lemma convex cone comono.item2}If $X,Y$ is a comonotonic pair, then any two elements of the convex cone \(C_{X,Y} \coloneqq \conv(\cone (\{X\} \cup \{Y\} ))\) are comonotonic to one other.
        \item \label{lemma convex cone comono.item3} Additionally, if neither $X$ or $Y$ are constants, then $C_{X,Y}\cap \R = \{0\}$.
    \end{enumerate}
\end{lemma} 
\begin{proof}
	In what follows, all equalities and inequalities are in the $\Pro\otimes\Pro$-almost sure sense, that is, they hold for any pair $(\omega,\omega')$ lying in an event $\Omega_1\subseteq \Omega\times\Omega$ having total $\Pro\otimes\Pro$ measure. $\Omega_1$ can be taken as the countable intersection of the events where the required inequalities (for any pairing of $X$, $Y$, $Y_n$, $Z$ and $W$) hold.
	
	We start proving \Cref{lemma convex cone comono.item1}. To see that $C_X$ is a cone, note that for any $Y \in C_X$ we have, by definition,
	\(
	\big(X(\omega) - X(\omega')\big) \left(Y(\omega)-Y(\omega') \right) \geq 0,
	\)
	for any $(\omega,\omega')\in \Omega_1.$ Hence, for any $\lambda \geq 0$ and {$(\omega,\omega')\in \Omega_1$},
	\[
	\big(X(\omega) - X(\omega')\big) \big(\lambda Y(\omega)- \lambda Y(\omega') \big)= \lambda \big(X(\omega) - X(\omega')\big) \big(Y (\omega)- Y (\omega') \big) \geq 0,
	\]
	yielding $\lambda Y \in C_X$. For convexity, let $Y,Z \in C_X$. Then, for $\lambda \in [0,1]$ we have that,
	\begin{align*}
	& \Big[ X(\omega) - X(\omega') \Big] \Big[ \big(\lambda Y(\omega) + (1-\lambda) Z(\omega)\big) - \big(\lambda Y(\omega') + (1-\lambda) Z(\omega')\big) \Big]
	\\ = &\lambda \left[ X(\omega) - X(\omega')\right] \left[ Y(\omega)-Y(\omega') \right] + (1-\lambda) \left[ X(\omega) - X(\omega')\right] \left[ Z(\omega)-Z(\omega') \right] \geq 0
	\end{align*}
	whenever $(\omega,\omega')\in \Omega_1$. {To see that $C_X$ is closed in the asserted sense, consider a convergent sequence $\{Y_n\}\subseteq C_X$ with $Y_n \to Y$ in probability. By standard facts of measure theory, there is a subsequence $\{Y_{n(k)}\}$ such that $Y_{n(k)}\to Y$ almost surely. Clearly, this yields that $Y$ is comonotonic to $X$.}
	
	For \Cref{lemma convex cone comono.item2}, let $Z,W \in C_{X,Y} $. By definition we have
	\(Z = \gamma_1 (\lambda_1 X) + (1-\gamma_1)(\delta_1 Y)\)
	for some triplet $(\gamma_1, \lambda_1, \delta_1)$ with $0\leq \gamma_1\leq1$ and $0\leq \lambda_1,\delta_1$, and similarly
	\(W = \gamma_2 (\lambda_2 X) + (1-\gamma_2)(\delta_2 Y)\)
	for some triplet $(\gamma_2, \lambda_2, \delta_2)$ with $0\leq \gamma_2\leq1$ and $0\leq \lambda_2,\delta_2$. Then, for $(\omega,\omega')\in \Omega_1$, expanding the product
	\[
	\big(Z(\omega) - Z(\omega')\big)\big(W(\omega) - W(\omega')\big)
	\]
	yields a weighted sum whose terms are all non-negative. 

For the last item, it is enough to verify that the additive combination of non-constants comonotonic random variables can not be constant. As $X$ is non-constant, then there is $\omega,\omega' \in \Omega$ such that $X(\omega) <X(\omega')$ and comonotonicity implies $Y(\omega) \leq Y(\omega')$. Therefore, for any $\alpha,\beta > 0$ it holds that $(\alpha X + \beta Y)(\omega) =  \alpha X (\omega)  + \beta Y (\omega) < \alpha X (\omega')  + \beta Y (\omega) \leq (\alpha X + \beta Y)(\omega')$. Hence, $\alpha X + \beta Y$ is not constant. 
\end{proof}

\begin{remark}
	Note that the set \(C \coloneqq \bigcap_{Y \in C_X} C_Y, \)
	 is a non-empty, closed, and convex set, such that all its elements are comonotonic to one another. In particular, $\R \subseteq C$ and $C + \R = C$.
\end{remark}

We are now in a position to prove the main result of this section.

   


\begin{theorem}\label{main.theorem}
Let $\mathcal{A}$ be a monetary acceptance set and $\rho$ a risk measure. Then we have the following:
\begin{enumerate}
    \item \label{main.1.a} If  $\mathcal{A} $ and $\mathcal{A}^{\complement}$ are comonotonic convex, then $\rho_{\mathcal{A}}$ is comonotonic additive. The converse implication holds if $\mathcal{A}$ is closed.
    \item \label{main.3} The risk measure $\rho$ is comonotonic additive if and only if $\mathcal{A}_{\rho} $ and $\mathcal{A}_{\rho}^{\complement}$ are comonotonic convex.
   \end{enumerate}
\end{theorem}
\begin{proof}

For the first part of \Cref{main.1.a}, let $X$ and $Y$ be a comonotonic pair. By \Cref{lemma convex cone comono}, all elements in $ \conv (\cone ( \{X\} \cup \{Y\}))$ are comonotonic to each other. This implies, in light of the comonotonic convexity of $\mathcal{A}$ and $\mathcal{A}^{\complement}$, that $\mathcal{A}\cap C_{X,Y} $ and $\mathcal{A}^{\complement}\cap C_{X,Y}$ are convex sets. Since $\conv (\cone ( \{X\} \cup \{Y\})) + \R  $ is stable under scalar addition, the result follows from \Cref{main.lemma}. The converse of \Cref{main.1.a} follows directly from the converse of \Cref{main.lemma.s.3} of \Cref{main.lemma}

Regarding the ``only if" part of \Cref{main.3}, we will  show that $\mathcal{A}_{\rho}^{\complement}$ is comonotonic convex. A similar argument also holds for  $\mathcal{A}_{\rho} $. Take a comonotonic pair $X,Y \in \A_\rho^\complement$. We need to show that $\lambda X + (1-\lambda)Y \in \A_{\rho}^\complement $ for any $\lambda \in [0,1]$. But $ \rho (\lambda X + (1-\lambda) Y ) = \lambda \rho (X) + (1-\lambda) \rho(Y) > 0 $, which concludes the proof. The converse direction follows directly from \Cref{main.1.a} and the fact that $\rho=\rho_{\mathcal{A}_\rho}$.
\end{proof}

\section{Deviation measures}\label{deviations}

For deviations, a similar line of reasoning applies as for monetary risk measures but with distinct technical machinery. To explore this further, we introduce additional properties that comprise the basic setup to study Minkowski deviation measures. 

\begin{definition}\label{def.accept2}
An acceptance set $\mathcal{A}$ is a \defin{Minkowski acceptance set} if it satisfies the following:  \begin{enumerate}[resume = axioms.A,label=\textbf{\Alph*.}]
  
\item ({Star-shapedness}) $ \mathcal{A} $ is \defin{star-shaped} if $\lambda X\in \mathcal{A}$, for every $X \in \mathcal{A}$ and $\lambda \in [0,1]$. 


\item({Stability under scalar addition}) $ \mathcal{A} $ is \defin{stable under scalar addition} if $\mathcal{A} + \R = \mathcal{A}$, that is, if $ X + c \in \mathcal{A}$, for all $X \in \mathcal{A}$ and $c \in \R$.

\item({Radial boundedness at non-constants}) $\mathcal{A}$ is \defin{radially bounded at non-constants} if, for every $X \in \mathcal{A} \backslash \R$, there is some $ \delta_X \in (0 , \infty)$, such that $ \delta X \notin \mathcal{A}$ whenever $\delta \in [\delta_X , \infty)$. 


 \end{enumerate}
\end{definition}


\begin{definition}\label{def.dev}
For a functional $\mathcal{D}\colon\X\to [0,+\infty]$ we define its sub-level sets of the form $\Acc{}{\mathcal{D}}\coloneqq \{X\in\X\colon\,\mathcal{D}(X)\le 1\}$. Further, $\mathcal{D}$ is a \textbf{deviation measure} if it fulfils: \begin{enumerate}
  
\item ({Non-negativity}) $\mathcal{D}$ is \defin{non-negative} if $\mathcal{D} (X) > 0$ for any non-constant $X\in \mathcal{X}$ and $\mathcal{D}(X) = 0 $ for any constant $X \in \mathcal{X}$.

\item ({Translation insensitivity}) $\mathcal{D}$ is \defin{translation insensitive} if $\mathcal{D} (X + c) = \mathcal{D} (X)$ for any $X\in\X$ and $c \in \R$.

\item (Positive homogeneity) $\mathcal{D}$ is \textbf{positive homogeneous} if $\mathcal{D}(\lambda X) = \lambda \mathcal{D}(X)$ for any $X\in\X$ and $ \lambda\geq 0 $.
\end{enumerate}
A deviation measure may also satisfy the properties in \Cref{def.risk.measure}.
\end{definition}

We now define the Minkowski Deviation, introduced in \cite{moresco2023minkowski}, which is the main tool used in this section. A financial interpretation is that such a map indicates how much we should shrink (or ``gauge'') a certain position for it to become acceptable.

\begin{definition}\label{induce deviation} \label{SID}
Let $ \mathcal{A} \subseteq \X .$ The \defin{Minkowski Deviation of $\mathcal{A}$} is the functional $\D_\mathcal{A}\colon\X \to [0,+\infty]$ defined, for $X\in\X,$ by
\begin{align}\label{eq:fA}
\D_\mathcal{A}(X) \coloneqq \inf \left\{m > 0\colon\, m^{-1}{X} \in \mathcal{A} \right\},
\end{align}
where $\inf \varnothing = +\infty .$ 
\end{definition}

In analogy to \Cref{lemma.basic}, the next lemma relates acceptance sets to Minkowski deviations.

\begin{lemma}[Theorem 3.2 and 3.5 of \cite{moresco2023minkowski}]\label{radially} \label{generalized dev}\label{lemma-item4}
    Let $\mathcal{D}$ be a deviation measure and let $\A$ be a Minkowski acceptance set. Then we have the following: 
    \begin{enumerate}
    
        \item $\mathcal{D}(X) = \D_{\Acc{1}{\D}}(X)$ for all $X \in \X$ 
        \item  $\{ X \in \X : \f (X) <1 \} \subseteq \A \subseteq \mathcal{A}_{\f} \subseteq \cl (\A)$, where $\cl(\A)$ is the closure of $\A$.
        \item  If $\mathcal{A}$ is convex, then $\f$ is convex. Conversely, if $\mathcal{D}$ is convex, then $\mathcal{A}_{\D}$ is  convex.
         \item\label{lemma dev 4}   $\f$ is a deviation measure and  $\mathcal{A}_{\D}$ is  a Minkowski acceptance set.
    \end{enumerate}
\end{lemma}



Now, we turn our focus to the main results of this section. Similarly to what we did in the previous section, we define an auxiliary map, which represents the most we can shrink a position while keeping it non-acceptable.

 \begin{definition}
 The cogauge of $\mathcal{A}^\complement$ is the functional $ \cog_{\mathcal{A}^\complement}\colon\X \rightarrow [0,+\infty]$ defined, for $X\in \X$, by
\begin{align}
\cog_{\mathcal{A}^\complement} (X) \coloneqq \sup \left\{m \in \R_+^*\colon\, m^{-1}{X} \in \mathcal{A}^\complement \right\},
\end{align}
where $\sup \varnothing = 0 $.
 \end{definition}

We have the following relation between gauge and co-gauge.

\begin{lemma}(Corollary C.8. of \cite{moresco22tese})\label{coro cogauge}
Let $\mathcal{A}\subseteq\X$ be star-shaped. Then $\f (X) = \cog_{\mathcal{A}^\complement} (X)$
holds for all $X\in\X$.
\end{lemma}

 We now prove a result regarding (sub/super) additivity of deviation measures, which will be very useful for the main result.

\begin{theorem}
    \label{additive}
Let $\mathcal{A}\subseteq\X$ be a Minkowski acceptance set and $\mathcal{D}$ a deviation measure. Then we have that:
\begin{enumerate}
\item \label{item1.additive}If $\mathcal{A}$ is convex, then $\f$ is sub-linear (convex and positive homogeneous).
 \item \label{item2.additive}If $\mathcal{A}^\complement$ is convex, then $\f$ is super-linear (concave and positive homogeneous) on $\cone (\mathcal{A}^\complement)$, that is, $\f(X + Y ) \geq \f(X) + \f(Y)$ for any $X,Y \in \cone (\mathcal{A}^\complement)$.
 \item \label{item3.additive} If $C \subseteq \cone (\mathcal{A}^\complement)$ is a cone for which both $\mathcal{A} \cap C $ and $\mathcal{A}^\complement \cap C $ are convex sets, then $\f $ respects $\f(X + Y) = \f(X) + \f(Y) $ for every $ X,Y \in C$.
 \item \label{item4.additive}If $\mathcal{D} $ is additive in some convex cone $C$, then $ \Acc{k}{\D} \cap C$ and $(\Acc{k}{\D})^\complement \cap C$ are convex sets.  
 \end{enumerate}
\end{theorem}

\begin{proof}
\Cref{item1.additive} follows from \Cref{radially}. For \Cref{item2.additive}, we already have positive homogeneity from \Cref{radially} \Cref{lemma dev 4}. The star-shapedness of $\mathcal{A}$ and \Cref{coro cogauge} tells us that $\f = \cog_{\mathcal{A}^\complement}$. Hence, it suffices to show that $\cog_{\mathcal{A}^\complement}$ is a concave functional on $\cone (\mathcal{A}^\complement)$ whenever $\mathcal{A}^\complement$ is convex. To see that this is the case, let $B = \mathcal{A}^\complement$, and fix $\lambda\in[0,1]$ and $X,Y\in \cone (\mathcal{A}^\complement)$.
	Let us first consider the case where $0<\lambda<1$ and where both $X$ and $Y$ are nonzero. In this scenario, the sets
	\(
	\mathfrak{A} \coloneqq \{\alpha\in\R_+^*\colon\, \lambda X\in \alpha B\}
	\)
	and
	\(
	\mathfrak{B} \coloneqq \{\beta\in\R_+^*\colon\, (1-\lambda)Y\in \beta B\}
	\)
	are both non-empty (for instance, $X\in\cone(B)$ means precisely that $X = aZ$ for some $a>0$ and some non-zero $Z\in B$, and in this case we have $\lambda a \in \mathfrak{A}$). The positive homogeneity of $\f$ together with the equality $\f = \cog_B$, implies that $\sup\mathfrak{A} = \cog_B(\lambda X) = \lambda \cog_B(X)$ and $\sup\mathfrak{B} = \cog_B((1-\lambda)Y) = (1-\lambda)\cog_B(Y)$. Taking $\alpha\in\mathfrak{A}$ and $\beta\in\mathfrak{B}$, convexity of $B$ yields $\lambda X+(1-\lambda)Y \in (\alpha + \beta)B$, so $\cog_B (\lambda X + (1-\lambda)Y) \geq \alpha + \beta$. Therefore, $\cog_B (\lambda X + (1-\lambda)Y) \geq \sup\mathfrak{A} + \sup\mathfrak{B} = \lambda \cog_B (X) + (1-\lambda)\cog_B (Y) $. The remaining cases are just a matter of adapting the following argument: if, say, $\lambda X = 0$, then $\mathfrak{A}=\varnothing$ and $\cog_B(\lambda X + (1-\lambda Y)) = \cog_B((1-\lambda)Y) = (1-\lambda)\cog_B(Y) = \lambda\cog_B(X) + (1-\lambda)\cog_B(Y)$. 

 Regarding \Cref{item3.additive}, let $g$ be the restriction of $\f$ to the cone $C$, i.e.,\ $g \colon C \rightarrow [0,\infty]$ is such that $g (X) = \f(X) = \max\big(\f(X), \D_C(X)\big) = \D_{\mathcal{A} \cap C} (X)$ for all $X \in C$. It suffices to show that $g$ is additive; we shall proceed by showing that this function is concave and sub-linear. Sub-linearity of $g$ is yielded as $\mathcal{A} \cap C $ is a convex set containing the origin by assumption (see Theorem 3.2 in \cite{moresco2023minkowski} -- item (v)). Therefore $\D_{\mathcal{A}\cap C}$ is sub-linear on the whole $\X$, in particular when restricted to $C$. For concavity, we shall summon the cogauge to help us: as $\mathcal{A} $ is a star-shaped set, the gauge coincides with the cogauge of its complement, i.e.,\ $ \f = \cog_{\mathcal{A}^\complement}$ --- see \Cref{coro cogauge}. It follows that, for $X\in C$, one has $g (X) = \cog_{\mathcal{A}^\complement} (X)$. We now show that, for $X\in C$, the identity $\cog_{\mathcal{A}^\complement} (X) = \cog_{\mathcal{A}^\complement \cap C} (X)$ holds. 
 As $C^\complement\cup\{0\} $ is a cone and any cone is star-shaped, $\mathcal{A}\cup C^\complement$ is star-shaped, then we have $\forall X \in \mathcal{X}$ that
	\begin{align*}
	    \cog_{\mathcal{A}^\complement \cap C} (X) = \cog_{(\mathcal{A}\cup C^\complement)^\complement} (X)
	&= \D_{\mathcal{A} \cup C^\complement} (X)\\
	&= \min(\f(X) , \D_{C^\complement} (X) )
	= \min \big(\cog_{\mathcal{A}^\complement}(X),\cog_{C}(X)\big).
	\end{align*}
	In particular, $g = \cog_{\mathcal{A}^\complement\cap C}$ on $C$, as $\cog_C (X) = \infty = \D_{C^\complement} (X)$ if $X \in C$ and $\cog_C (X) = 0=\D_{C^\complement} (X)$ if $ X \notin C$. Now, the only thing that is left to show is that the cogauge of a convex set is a concave function on $C$. This claim follows from \Cref{item2.additive} as it tells us that $\cog_{\mathcal{A}^\complement \cap C}$ is concave on $C\subseteq\cone (\mathcal{A}^\complement)$.

 For \Cref{item4.additive}, note that the restriction of $\mathcal{D}$ to $C$ is both convex and concave. Therefore, the convexity of both $\Acc{k}{\D} \cap C$ and $(\Acc{k}{\D})^\complement \cap C$ follows from Theorem 3.7 in \cite{moresco2023minkowski} -- item (v) and \Cref{coro cogauge}. 
\end{proof}

\begin{remark}
    \Cref{item4.additive} in the above Theorem can easily be relaxed to the following: if $\mathcal{D}$ is sub-(super-)additive in some convex cone $C$, then $ \Acc{k}{\D} \cap C$  ($(\Acc{k}{\D})^\complement \cap C$, respectively) is a convex set. Unfortunately, \Cref{item2.additive,item3.additive} of \Cref{additive}  cannot be relaxed so as to accommodate the superlinearity of $\f$ on the whole domain. Consider the following counterexample,  illustrated in \Cref{fig counter concave}:
 let $\Omega = \{0,1\}$ be the binary market and identify $L^0\equiv\R^2$ as usual. Let $\mathcal{A}\coloneqq\{(x,y)\in\R^2\colon\, y-\vert x\vert\leq1\}$. In this case, the set $C\coloneqq \mathcal{A} \setminus \cone (\mathcal{A}^\complement)$ is a cone and hence, for any $X \in C$, we have that $\f(X) = 0$, whereas $\f(X)>0$ for $X\notin C$. We denote, respectively, by $\bd C$ and $\interior C$  the interior and boundary of $C$. Now let $Y = (1,\nicefrac12)\in \interior C$, $Z = (1,1)\in \bd C$ and $W = (1,2)\in\bd \mathcal{A}$. We have 
	$\f(Z) =0< \f(W)$, but $Z$ is a convex combination of $W$ and $Y$, so $\f$ is not concave on the whole domain.
	 However, if we are willing to abandon the identity $\f = \cog_{\mathcal{A}^\complement}$, it is possible to define the cogauge in a slightly different way by assigning the value $\cog_B(X)\coloneqq -\infty$ whenever $\{m \in \R_+ \colon\, m^{-1} X \in B\} = \varnothing $; in this case, an easy adaptation yields the concavity of $\cog_B$ for convex $B$. 
	\end{remark}

\begin{figure}[h!]
	\caption{A star-shaped set $A$ (in gray) with convex complement for which $\f$ is not concave.}\label{fig counter concave}
	\begin{center}
		\begin{tikzpicture}[scale=1]
		
		\clip (0,0) (-4,-2)rectangle (4,4);
		\fill[fill=black!20, fill opacity=.8, rotate = 0](0,0) (-4,-4)rectangle (4,4);
		\filldraw[ fill=white!20, rotate = 0] (-4,5)--(0,1)-- (4,5);
		
		\filldraw[ fill=red!20, fill opacity=.50, dashed] (-5,5)--(0,0)-- (5,5);
		
		\fill (1,3) circle [radius=.0cm] node[anchor=south west,scale=.8]{$\A^\complement$};
		\fill (1,-1) circle [radius=.0cm] node[anchor=south west,scale=.8]{$\A$};
		\fill (2.8,3.5) circle [radius=.0cm] node[anchor=north,scale=.8, rotate=45]{$\cone (\A^\complement)$};
		
		\fill (0,0) circle [radius=.05cm] node[anchor=north west,scale=.8]{$0$};
		
		\fill (1,1) circle [radius=.05cm] node[anchor=west,scale=.8]{$Z$};
		\fill (1,2) circle [radius=.05cm] node[anchor=west,scale=.8]{$W$};
		\fill (1,0.5) circle [radius=.05cm] node[anchor=west,scale=.8]{$Y$};
		
		\draw[thin,<->] (-4,0) -- (4,0);
		\draw[thin,<->] (0,-2) -- (0, 4);
		\end{tikzpicture}
	\end{center}
\end{figure}
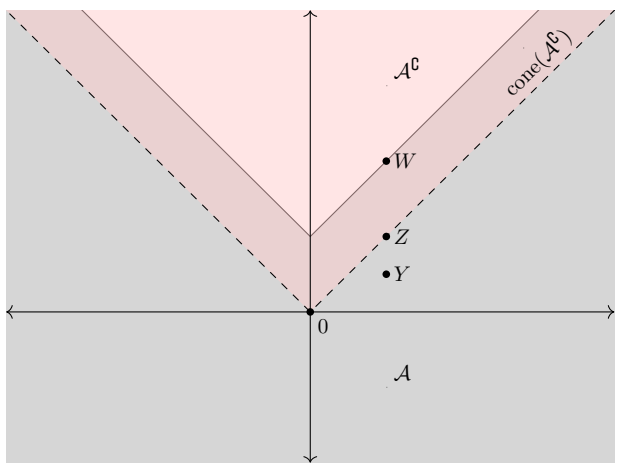

We are now in condition to prove the main result in this section.

\begin{theorem}\label{coro como}
We have the following:
\begin{enumerate}
    \item \label{item1.como} Consider an acceptance set $\mathcal{A}\subseteq\X$ being radially bounded at non-constants, stable under scalar addition, and assume both $\mathcal{A}$ and $\mathcal{A}^\complement$ be comonotonic convex. Then $\mathcal{A}$ is star-shaped and $\f$ a comonotonic additive deviation measure. 
  \item \label{item2.como} Let $\mathcal{D}$ be a deviation measure that is comonotonic additive. Then both $\Acc{k}{\D}$ and $\Acc{k}{\D}^\complement$ are comonotonic convex.  
\end{enumerate}
\end{theorem}

\begin{proof}
For \Cref{item1.como}, the star-shapedness of $\mathcal A$ follows for $0$ is comonotonic to any $X \in \mathcal{A}$ and, by assumption, $\mathcal{A}$ is convex for this pair. Therefore, $\lambda X\equiv \lambda X + (1-\lambda) 0 \in \mathcal{A}$ for all $X \in \mathcal{A}$ and any $0\leq\lambda\leq1$, which establishes star-shapedness. Furthermore, as $\mathcal{A}$ is radially bounded at non-constants, it follows that $\cone (\mathcal{A}^\complement) = ( \X\setminus \R) \cup \{0\}$ and so any cone with no constants that we may take is contained in $\cone(\mathcal{A}^\complement)$. Now let $X$ and $Y$ be a comonotonic pair of non-constants. Note that any two members of the set $C_{X,Y} = \conv (\cone (\{X\} \cup \{Y\} ))$ are comonotonic to one another and the only constant in $C_{X,Y}$ is $0$ (see \Cref{lemma convex cone comono}). Now, if we take any $Z,W \in C_{X,Y} \cap \mathcal{A}$, as they are a comonotonic pair, by assumption we have that $\lambda Z + (1-\lambda)W \in C_{X,Y} \cap \mathcal{A},\; \forall \lambda \in [0,1]$. Hence, $ C_{X,Y} \cap \mathcal{A}$ is a convex set. The same argument shows that $ C_{X,Y} \cap \mathcal{A}^\complement$ is also convex. Thus, by  \Cref{additive}, we have that $\f(X+Y) = \f(X) + \f(Y)$. To conclude the first item, notice that $\f$ is a deviation measure because $\mathcal{A}\subseteq\X$ is a Minkowski acceptance set (\Cref{radially}).

For \Cref{item2.como}, let $X,Y$ be a comonotonic pair. Due to Lemma \ref{lemma convex cone comono}, the set $C_{X,Y}$
is a convex cone whose members are all comonotonic to one another, and $\mathcal{D}$ is additive on $C_{X,Y}$. By \Cref{additive} \Cref{item4.additive}, the sets $\Acc{k}{\D}\cap C_{X,Y}$ and $(\Acc{k}{\D})^\complement \cap C_{X,Y}$ are both convex. In particular, if $Z$ is any convex combination of $X$ and $Y$, then $Z\in \Acc{k}{\D}\cap C_{X,Y}\subseteq \Acc{k}{\D}$ whenever $X,Y \in \Acc{k}{\D}$, and similarly $Z\in (\Acc{k}{\D})^\complement$ whenever $X,Y \in (\Acc{k}{\D})^\complement$.
\end{proof}

\begin{remark}
	If the conditions above are imposed only on $\mathcal{A}$ (and not necessarily on $\mathcal{A}^\complement$), then we have that $\f$ is comonotonic convex. Similarly, if we only impose those conditions on $\mathcal{A}^\complement$, then the resulting $\f$ is comonotonic concave. The converse implications also hold. As an example of a set $\mathcal{A}$ satisfying the assumptions in the theorem, take $\Omega = \{0,1\}$, identify $L^0\equiv\R^2$, and let $\mathcal{A}$ be the set of those $X=(u,v)\in\R^2$ for which $u\geq0$, $v\geq0$ and $\vert u\vert + \vert v\vert \leq 1$. In this case, the set of comonotonic pairs in the 1st quadrant is precisely $\{(u,v)\in\R_+^2\colon\,u\geq v\}$.
\end{remark}

\bibliography{paper.bib}
\bibliographystyle{apalike}

\end{document}